\documentclass[12pt,a4paper]{article}
\usepackage{t1enc}
\usepackage[latin2]{inputenc}
\usepackage[mathscr]{eucal}
\usepackage[reqno,intlimits]{amsmath} 
\usepackage{amssymb,amsthm}

\usepackage[a4paper,left=2.5cm,right=2.5cm,top=2.5cm,bottom=3cm]{geometry}

\usepackage{graphicx}   

\usepackage{makeidx}

\numberwithin{equation}{section} 

\usepackage{enumerate}

\newcommand*{\ba}{\begin{array}}
\newcommand*{\ea}{\end{array}}
\newcommand*{\be}{\begin{equation}}
\newcommand*{\ee}{\end{equation}}
\newcommand*{\bfe}{\renewcommand{\theenumi}{\thechapter.\arabic{enumi}}\begin{enumerate}}
\newcommand*{\efe}{\end{enumerate}\renewcommand{\theenumi}{\arabic{enumi}}}
\newcommand*{\ben}{\begin{enumerate}}
\newcommand*{\een}{\end{enumerate}}
\newcommand*{\bit}{\begin{itemize}}
\newcommand*{\eit}{\end{itemize}}

\newcommand*{\Ev}{{\bf E}}
\newcommand*{\Var}{{\bf Var}}

\newcommand*{\pd}{\partial}

\newcommand*{\inner}[2]{\left<#1,\,#2\right>}

\newtheorem{defi}{Definition}

\newtheorem{tetel}{Theorem}
\newtheorem{all}[tetel]{Proposition}
\newtheorem{lem}[tetel]{Lemma}

\newtheorem{pelda}[tetel]{Example}
\newtheorem{sejtes}[tetel]{Conjecture}

\newtheorem*{megj}{Remark}

\newcommand{\R}{\mathbb{R}}

\newcommand{\Z}{\mathbb{Z}}
\newcommand{\C}{\mathbb{C}}

\title{Pauli channel tomography with unknown channel directions}
\author{D\'aniel Virosztek$^2$, L\'aszl\'o Ruppert$^{1,2}$, Katalin M. Hangos$^1$\vspace*{0.5cm}\\
$^1$ Process Control Research Group, \\ Computer and Automation Research Institute,\\ H-1518 Budapest, POB 63, Hungary\vspace*{0.25cm}\\
$^2$ Department of Analysis, \\ Budapest University of Technology and Economics, \\ H-1521 Budapest, POB 91, Hungary}
\date{}

\begin{document}
\maketitle

\let\thefootnote\relax\footnotetext{E-mail: virosz89@gmail.com, ruppertl@math.bme.hu, hangos@scl.sztaki.hu\\
This research was supported in part by the Hungarian Research Fund through grant $K83440$ and by the New Hungary Development Plan (Project ID: T\'{A}MOP-4.2.2.B-10/1-2010-0009).}

\begin{abstract} 
In this paper we estimate the parameters of the qubit Pauli channel using the channel matrix formalism. The main novelty of this work is that we do not assume the directions of the Pauli channel to be known, but they are determined through the tomography process, too. The results show that for optimally estimating the contraction parameters and the channel matrix we should have input qubits and measurements in the channel directions. However, for optimally estimating the channel directions, we should use different tomography conditions.
\end{abstract}

quadratic error, qubit.

\section{Introduction}
The accurate description of different quantum phenomena is a key issue in their potential use
 in modern IT-applications. In quantum mechanics, both dynamical changes and
communication is treated using quantum channels. Therefore the parameter estimation of quantum channels plays a major role in quantum information processing, and the area of quantum process tomography is flourishing \cite{Anis-2012, Bendersky-2008, Keshari-2011, Roga-2010, Teo-2011}.

Direct quantum process tomography is performed by sending known quantum systems into the  
channel, and then estimating the output state. In quantum mechanics the measurement has a probabilistic nature \cite{Nielsen-book,Petz-book}, therefore  many identical copies of the input quantum system are needed, and an estimator is constructed by using statistical considerations. For achieving efficient process tomography, experiment design is necessary that consists of selecting the optimal input state, optimal measurement of the output state, and an efficient estimator of the channel from the measured data. 

The field of quantum process tomography is well-established, an exhaustive description of possible tomography methods can be found in \cite{Mohseni-2008}. The Pauli channels form a relatively wide family of quantum channels. The tomography of Pauli channels has a huge literature, however, due to the level of difficulty of the topic, papers mostly deal with special cases, e.g., with the optimal parameter estimation of a depolarizing channel \cite{Sasaki-2002}. But there are some publications investigating the estimation of multi-parameter channels \cite{Bendersky-2008,Young-2009}, and the multidimensional case also appears \cite{Fujiwara2003, Nathanson-2007}. There are also some experimental results concerning the optimal estimation of the Pauli-channels \cite{Branderhorst2009,Chiuri-2011}.

In contrast to the majority of the works in this area, we propose an extended problem statement: we investigate qubit Pauli channels with unknown channel directions. Despite of the novelty of the approach, there are a few papers that deal with optimally estimating qubit Pauli channels including their channel directions. In \cite{Ballo-2011} the problem was examined using convex optimization methods, and a numerical method was provided for finding the optimal input - measurement pairs. In \cite{Ruppert-2012} we examined the optimality of the estimation problem using purely statistical considerations to achieve analytical results. However, analytical results could only be obtained for the case of known channel directions. 

Therefore, the aim of this paper is to give an analytical description of the optimal estimation of Pauli channels in the case of unknown channel directions, too.
The efficiency of these estimations is measured here with three quantities: the mean squared error of the estimated contraction parameters and angle parameters, and the mean distance of the estimated and the real channel matrix are investigated. 

The paper is organized as follows. 
In Section 2 we introduce the necessary notions to understand the rest of the article. In Section 3 the tomography method used is described. In Section 4 the optimization of the previously mentioned quantities are performed. Finally, conclusions are drawn.

\section{Preliminaries}
 
In the following section, we give a short introduction to the applied concepts. A more detailed description can be found in \cite{Nielsen-book, Petz-book}. We will only examine two-level systems, that is quantum bits or qubits.

\medskip
The state of two-level quantum systems is described by \textit{density matrices} $\rho \in \mathbf{M}_{2}(\C)$ that are parametrized by a real vector $\theta \in \mathbb{R}^{3}$ 
called the \textit{Bloch vector}.  
\begin{defi}[Bloch parametrization]
\be \label{bloch}
\rho: \mathbb{R}^3 \rightarrow\mathbf{M}_{2}(\C); \
\theta=(\theta_1, \theta_2, \theta_3)^T \mapsto \rho(\theta)= \frac{1}{2} \left(I +\theta_1 \sigma_1+\theta_2 \sigma_2+\theta_3 \sigma_3, \right)
\ee
where
\[
I=\sigma_0 = 
\left( \begin{array}{cc} 
1 & 0 \\
0 & 1 \\
\end{array} \right),
\
\sigma_1 = 
\left( \begin{array}{cc} 
0 & 1 \\
1 & 0 \\
\end{array} \right),
\
\sigma_2 = 
\left( \begin{array}{cc} 
0 & -i \\
i & 0 \\
\end{array} \right),
\
\sigma_3 = 
\left( \begin{array}{cc} 
1 & 0 \\
0 & -1 \\
\end{array} \right).
\]
\end{defi}
\begin{defi}[Quantum state]
A qubit can be described with $2 \times 2$ density matrices ($\rho(\theta)$) satisfying the following condition:
\begin{align*}
\mathrm{Tr}(\rho(\theta))&=1
\\
\rho(\theta) &\geq 0.
\end{align*}
\end{defi}
It is easy to check that $\rho(\theta)$ is a density matrix if and only if $\theta_1^2+\theta_2^2+\theta_3^2 \leq 1$.
That is, the state space can be represented with the unit ball in $\R^3$, the so-called Bloch ball.


\medskip
We will use \textit{von Neumann measurements} with two possible outcomes \cite{PR-2012}.
\begin{defi}[Measurement]
$\{P,I-P\}$ is a von Neumann measurement if $P$ is $2 \times 2$ projection. The probability of measuring outcome $P$ on the system $\rho(\theta)$ is $\mathrm{Tr} (\rho(\theta) P)$.
\end{defi}

The only non-trivial case is when $P$ is a rank-one projection. Then using the same Bloch parametrization as previously, we can rewrite $P$ in the following form:
\be \label{bl-par-meas}
P=\frac{1}{2}\left( \sigma_0+m_1 \sigma_1+m_2 \sigma_2+m_3 \sigma_3 \right),
\ee
where $m_1^2+m_2^2+m_3^2=1$. Using the abbreviation $\underline{m}=(m_1,m_2,m_3)^T$ and $\underline{\theta}=(\theta_1,\theta_2,\theta_3)^T$ we obtain
\[
\mathrm{Prob}(\text{"measuring P"})=\mathrm{Tr}\big(\rho(\theta) P\big)=\frac{1}{2}(1+\underline{m} \cdot \underline{\theta}).
\]


\medskip
\textit{Quantum channels} are completely positive, trace-preserving maps, and Pauli channels are a well-known family of them in the qubit case.
\begin{defi}[Pauli channel]
Let be $\{v_0=I, v_1, v_2, v_3 \}$ an arbitrary base satisfying $\mathrm{Tr} (v_i v_j)=2\delta_{i,j}~, v_i \in \mathbf{M}_{2}^{s.a.}(\C) ~ \forall i,j \in \{0,1,2,3\}$. Let be $\lambda_1, \lambda_2, \lambda_3$ real numbers that fulfill
\be \label{cpfelt}
1 \pm \lambda_3 \geq |\lambda_1 \pm \lambda_2|.
\ee
Then a Pauli channel can be described with a mapping

\be \label{pauli_hat}
\mathcal{E}:\mathbf{M}_{2}(\C) \rightarrow \mathbf{M}_{2}(\C); \ 
\rho=\frac{1}{2}\left(I + \sum_{i=1}^3 \theta_i v_i \right)
\mapsto \mathcal{E}(\rho)=\frac{1}{2}\left(I + \sum_{i=1}^3 \lambda_i \theta_i v_i \right).
\ee
The affine subspaces $ \{\frac{1}{2}\left(I+t v_i \right): t \in \R \} \subset \mathbf{M}_{2}^{s.a.}(\C)$ ($i \in \{1,2,3\}$) are called the channel directions, the numbers $\lambda_1, \lambda_2, \lambda_3$ are called the contraction parameters.
\end{defi}

The completely positiveness of the mapping is guaranteed by condition (\ref{cpfelt}) (see \cite{Petz-2008}). From (\ref{cpfelt}) follows that $-1 \leq \lambda_i \leq 1, \ i \in \{1,2,3\}$. If we look at the effect of a Pauli channel on the Bloch vectors then it becomes clear that we have contractions with parameter $\lambda_i$ in the appropriate channel directions. In other words, the image of the whole Bloch ball will be an ellipsoid with its axes lying in the directions of the channel and being $2 |\lambda_i|$ long.

\begin{defi}[Channel matrix]
We will call the mapping $A: \mathbb{R}^3 \rightarrow \mathbb{R}^3$ the channel matrix of Pauli channel $\mathcal{E}$, if
\be \label{kommut}
\mathcal{E} \circ \rho = \rho \circ A.
\ee
\end{defi}

Straightforward calculations show that $A$ is a linear mapping and it can be always parametrized in the following form:
\begin{tetel}\label{alap_tetel}
The channel matrix $A$ of every Pauli channel can be constructed as
\be
\label{alap}
A( \lambda_1,\lambda_2, \lambda_3,\phi_z, \phi_y, \phi_x)=R_{z} R_{y} R_{x} \Lambda R_{x}^{-1} R_{y}^{-1} R_{z}^{-1},
\ee
where
\begin{align*}
R_{z} (\phi_z)&=
\left(
\begin{array}{ccc}
\cos{\phi_z} & -\sin{\phi_z} & 0 \\
\sin{\phi_z} & \cos{\phi_z} & 0 \\
0 & 0 & 1 
\end{array}
\right), 
\
&&R_{y} (\phi_y)=
\left(
\begin{array}{ccc}
\cos{\phi_y} & 0 & -\sin{\phi_y} \\
0 & 1 & 0 \\
\sin{\phi_z} & 0& \cos{\phi_y}
\end{array}
\right),
\\
R_{x} (\phi_x)&=
\left(
\begin{array}{ccc}
1 & 0 & 0\\
0 & \cos{\phi_x} & -\sin{\phi_x} \\
0 & \sin{\phi_x} & \cos{\phi_x}
\end{array}
\right),
\
&&\Lambda=
\left(
\begin{array}{ccc}
\lambda_{1} & 0 & 0\\
0 & \lambda_{2} & 0 \\
0 & 0 & \lambda_{3}
\end{array}
\right).
\end{align*}
\end{tetel}

Note that this parametrization is only surjective, not bijective. This means that we can get the same channel matrix for several set of parameters.


\section{Tomography of Pauli channels}

Using the notions described in the previous section, we can characterize the Pauli channel with its channel matrix ($A$). Furthermore, Theorem \ref{alap_tetel} shows that the channel matrix can be constructed by using three contraction parameters ($\lambda_i$) and three angle parameters ($\phi_i$).
Our aim is to give the best estimation of these quantities. 

For this purpose we will send some input qubits through the channel, perform some measurements and construct an estimator from the measured data. 

\paragraph{Input qubits and measurements}
We have to have at least three different input qubits for complete channel tomography. Let us suppose that we have three different measurements, too. Previous investigations showed that one should choose pure input states, which are orthogonal to each other  \cite{Ballo-2011},\cite{Ruppert-2012}.

Let us suppose that the Bloch vectors of the input states are $\underline{\theta}^{(1)},\underline{\theta}^{(2)},\underline{\theta}^{(3)}$. If they are orthogonal, they create an orthogonal matrix and we can parametrize them the following way
\be \label{inpdef}
\Theta=\left[ \underline{\theta}^{(1)},\underline{\theta}^{(2)},\underline{\theta}^{(3)} \right]=R_z(\vartheta_z)R_y(\vartheta_y)R_x(\vartheta_x),
\ee
where $R_z, R_y, R_x$ are the same rotations as in (\ref{alap}) and $0 \leq \vartheta_z, \vartheta_y < \pi$, $0 \leq \vartheta_x < \frac{\pi}{2}$.

Similarly, we suppose that the measurements with Bloch vectors $\underline{m}^{(1)},\underline{m}^{(2)},\underline{m}^{(3)}$ are orthogonal  \cite{Ruppert-2012}. Thus, we can parametrize them using one-dimensional rotations, too:
\be \label{merdef}
M=\left[ \underline{m}^{(1)},\underline{m}^{(2)},\underline{m}^{(3)} \right]=R_z(\tau_z)R_y(\tau_y)R_x(\tau_x)
\ee
where $0 \leq \tau_z, \tau_y < \pi$, $0 \leq \tau_x < \frac{\pi}{2}$.


\paragraph{Estimation of the channel matrix}

The matrix generated from the Bloch vectors of the output qubits $\Xi=\left[ {\underline{\xi}^{(1)}},{\underline{\xi}^{(2)}},{\underline{\xi}^{(3)}} \right]$ fulfill

\be \label{csathat}
\Xi=A( \lambda_1,\lambda_2, \lambda_3,\phi_z, \phi_y, \phi_x) \Theta.
\ee

Let us suppose that we perform the $i$-th measurement $N$ times on the copies of the $j$-th output qubit and denote by $N_{ij}$ the number of measurement results corresponding to $m^{(i)}$. Then $N_{ij}$ is binomially distributed with the following parameters
\be \label{nij2}
\ N_{ij} \sim \mathrm{Binom}\left(N, \frac{1+\underline{m}^{(i)} \cdot \underline{\xi}^{(j)} }{2} \right).
\ee
We can use the notation $x_{ij}=\underline{m}^{(i)} \cdot \underline{\xi}^{(j)}, \ \forall \ i,j \in \{1,2,3\}$, then $X=\{x_{ij}\}_{i,j=1}^{3}$ can be written in the form
\be \label{xmeghat}
X=M^{T}\Xi.
\ee

From (\ref{nij2}) we can estimate the elements of $X$ (since $\Ev (N_{ij})=N\cdot \frac{1+x_{ij}}{2}$):
\be \label{hatxijdef}
\hat x_{ij}:=\frac{2}{N}\cdot  N_{ij}-1.
\ee
Furthermore, the elements of $\Xi$ can also be estimeted from (\ref{xmeghat}) as
\be \label{koord}
\hat \Xi := M \hat X.
\ee
Finally, from (\ref{csathat}) we can obtain an estimation of the channel matrix
\be \label{hatadef}
\hat{A}:= \hat \Xi \Theta^{-1}=M \hat X \Theta^{T}.
\ee

\paragraph{Estimation of the channel parameters}

From (\ref{alap}) we know the mapping from the sets of parameters to the channel matrices:
\[
A: \mathcal{D} \rightarrow \mathbf{M}_3 (\R); \ (\lambda_1,\lambda_2, \lambda_3, \phi_z, \phi_y, \phi_x) \mapsto A(\lambda_1,\lambda_2, \lambda_3, \phi_z, \phi_y, \phi_x),
\]
where $\mathcal{D}  \subset \R^6 $, and our aim is to define its inverse mapping, i.e., we want to estimate the channel parameters from the estimation of the channel matrix (\ref{hatadef}). That is we need to find the mapping
\[
T: \mathbf{M}_3 (\R) \rightarrow \R^6; \ \hat A \mapsto (\hat \lambda_1, \hat \lambda_2, \hat \lambda_3, \hat \phi_z, \hat \phi_y, \hat \phi_x),
\]
which fulfills
\be \label{balinv2}
T \circ A=\mathrm{Id}_{\mathcal{D}}.
\ee


We can construct $T$ the following way.
Let us symmetrize our estimation: 
\[
\hat{A}_s:= \frac{1}{2} \left( \hat A + \hat{A}^{T} \right).
\]
Then $\hat{A}_s$ can be diagonalized in an orthonormal basis. The Jordan decomposition of $\hat{A}_s$ can be calculated easily algebraically since its characteristic polynomial is cubic. This way we get the eigenvalues $\hat \lambda_1 \geq \hat \lambda_2 \geq \hat \lambda_3$, which are exactly the estimates of the channel contraction parameters, and the corresponding normalized eigenvectors: $\mathbf{v_1}, \mathbf{v_2}, \mathbf{v_3}$ are also obtained.  

The next question is how we can calculate the angle parameters from these eigenvectors.
From (\ref{alap}) we can see the geometrical meaning of the angle parameters: $\phi_z$ and $\phi_y$ are the polar and azimuth angles of $\mathbf{v_1}$, while $\phi_x$ is the angle of $\mathbf{v_2}$ and the intersection of planes $z=0$ and $\mathbf{v_1}^{\perp}$ (orthogonal subspace of $\mathbf{v_1}$).

These quantities can be uniquely determined only if  $\lambda_1 > \lambda_2 > \lambda_3$, so we need to restrict the domain of mapping $A$ ($\mathcal{D}$) to get a one-to-one correspondence between channel matrices and channel parameters.

\begin{tetel}\label{domain}
The parametrization will be bijective on the following domain:
\[
\mathcal{D}=\{ (\lambda_1, \lambda_2, \lambda_3, \phi_z,\phi_y,\phi_x) \in \R^6 \ : ~ 1 \pm \lambda_3 \geq |\lambda_1 \pm \lambda_2|, \ \lambda_1 \geq \lambda_2 \geq \lambda_3,
\]
\[
\phi_z, \phi_y, \phi_x \in [0,\pi),
\phi_y=\frac{\pi}{2} \Rightarrow \phi_z=0, \ 
\lambda_1=\lambda_2=\lambda_3 \Rightarrow \phi_z=\phi_y=\phi_x=0,\]
\[\lambda_1=\lambda_2>\lambda_3 \Rightarrow (\phi_x=0 \text{ and } \phi_y=0 \Rightarrow \phi_z=0),\]
\[
\lambda_1>\lambda_2=\lambda_3 \Rightarrow \phi_x=0 \}.
\]
\end{tetel}


\section{Optimal input states and measurements}

Let us introduce some abbreviations:
\begin{align*}
\underline{\lambda}&=[\lambda_1,\lambda_2, \lambda_3] ~ - ~ \textrm{channel contraction parameters},
\\
\underline{\phi}&=[\phi_z, \phi_y, \phi_x]  ~ -~ \textrm{channel angle parameters}, 
\\
\underline{\tau}&=[\tau_z, \tau_y, \tau_x]   ~ -~ \textrm{measurement parameters},
\\
\underline{\vartheta}&=[\vartheta_z, \vartheta_y, \vartheta_x]  ~ -~ \textrm{input qubits parameters}. 
\end{align*}

In the previous section we gave an estimation method of the channel matrix ($A$), contraction parameters ($\underline{\lambda}$) and angle parameters ($\underline{\phi}$), now we want to analyze the estimates. 

The distance between the real and estimated parameters 
 can be given in various ways \cite{Belavkin-2005}. Let us define our loss functions as the averaged squared errors of the previously mentioned quantities:
\begin{align}\label{eff3}
f(\underline{\lambda},\underline{\phi},\underline{\tau},\underline{\vartheta},N)&=\Ev \left( \|A(\hat \lambda_1, \hat \lambda_2, \hat \lambda_3, \hat \phi_z,\hat \phi_y, \hat \phi_x)-A(\lambda_1, \lambda_2, \lambda_3, \phi_z, \phi_y, \phi_x) \|^{2} \right),
\\ \label{eff2}
g(\underline{\lambda},\underline{\phi},\underline{\tau},\underline{\vartheta},N)&=\Ev \left((\hat \lambda_1- \lambda_1)^2+(\hat \lambda_2- \lambda_2)^2+(\hat \lambda_3- \lambda_3)^2 \right),
\\ \label{eff1}
h(\underline{\lambda},\underline{\phi},\underline{\tau},\underline{\vartheta},N)&=\Ev \left(\mathrm{dist}(\hat \phi_z,\phi_z)^2+\mathrm{dist}(\hat \phi_y,\phi_y)^2+\mathrm{dist}(\hat \phi_x,\phi_x)^2 \right),
\end{align}
where $\|\cdot\|$ is the  Hilbert--Schmidt norm, $\mathrm{dist}(\hat \phi,\phi):= \mathrm{inf}\{|\hat \phi-(\phi + k \pi)|: \ k \in \Z\}.$

\begin{megj}
The estimation error of the channel matrix is equivalent to the average estimation error of output qubits, since
\begin{align*}
& \int_\theta \| \mathcal{E}(\rho(\theta))-\hat{\mathcal{E}}(\rho(\theta)) \|^2 d\theta=\int_\theta \frac12 \| A~ \theta- \hat A~\theta \|^2  d\theta= \frac12 \int_\theta \sum_{i=1}^3 \left( \sum_{j=1}^3   (A_{ij}-\hat A_{ij})  \theta_j \right)^2  d\theta=
\\
&=\frac12 \int_\theta \sum_{i=1}^3 \sum_{j=1}^3   \theta_j^2 (A_{ij}-\hat A_{ij})^2  d\theta + \frac12 \int_\theta \sum_{i=1}^3 \sum_{j < k}   2 \theta_j \theta_k (A_{ij}-\hat A_{ij}) (A_{ik}-\hat A_{ik})  d\theta=
\\
&=\frac12  \sum_{i=1}^3 \sum_{j=1}^3 \left(\int_\theta  \theta_j^2 d\theta\right) (A_{ij}-\hat A_{ij})^2 + \sum_{i=1}^3 \sum_{j < k}    \left(\int_\theta  \theta_j \theta_k d\theta\right) (A_{ij}-\hat A_{ij}) (A_{ik}-\hat A_{ik})  =
\\ 
&= c \|A-\hat A\|^2,
\end{align*}
where in the last step we used the symmetry of Bloch-ball: $\int_\theta  \theta_j^2 d\theta=$constant, $\int_\theta  \theta_j \theta_k d\theta=0$.
\end{megj}

In the following we will find the optimal inputs and measurements for a given channel ($\underline{\lambda},\underline{\phi}$) and number of measurements ($N$) that minimize the above loss functions.

\subsection{The optimal estimation of the channel matrix}

Before obtaining the optimal estimation settings, we will introduce some useful statements.

\begin{all}[Rotational invariance] \label{rotation}
For any $O \in \mathbf{M}_3 (\R)$ orthogonal matrix, the estimation of the Pauli channel described by the channel matrix $A$ using the input and measurement settings $\Theta$ and $M$ (see (\ref{inpdef}), (\ref{merdef})) is exactly as efficient as the estimation of the Pauli channel described by $O A O^{-1}$ with the input and measurement settings $O \Theta$ and $O M$.
\end{all}

 Therefore, it is enough to investigate Pauli channels with channel parameters $\phi_z=\phi_y=\phi_x=0.$

\begin{lem}[Unbiasedness 1]\label{unbias1}
The estimation of the channel matrix given in (\ref{hatadef}) is unbiased.
\end{lem}

\begin{proof}
The estimation of the elements of $X$ defined in (\ref{hatxijdef}) is unbiased (that is, $\Ev \hat X=X$). 
The estimator of (\ref{hatadef}) is a linear function of $\hat X$, the expected value is linear, too, so $\Ev \hat A=A$.
\end{proof}

Let us introduce the notation $\hat a_{ij}=[\hat A]_{ij} \ (i,j \in \{1,2,3\})$ and the index set $H=\{11, 12, \dots, 32, 33 \}.$
$\hat A$ is a linear transformation of $\hat X$, so 
\be \label{ckldef}
\hat a_{k}= \sum_{l \in H} c_{kl}(\underline{\tau},\underline{\vartheta}) \hat x_{l} \ (k \in H), 
\ee
where the constants come from the actual values of $\Theta$ and $M$, hence determined by the parameters $\underline{\vartheta}$ and $\underline{\tau}$.

\begin{lem}\label{linear}
Set $\psi=\sum_{k \in H} d_k \hat a_k \ (d_k \in \R).$ Then
\be \label{11lem}
\Var \left( \psi \right)=\sum_{l \in H} \left( \sum_{k \in H} d_k c_{kl}\left(\underline{\tau}, \underline{\vartheta} \right) \right)^2 \frac{1-x_l^2}{N}.
\ee
\end{lem}

\begin{proof}
From (\ref{nij2}) and (\ref{hatxijdef}) it is easy to calculate the distribution of the elements of $\hat X.$
Using the well known properties of the binomial distribution and the independence of the different elements of $\hat X,$ straightforward calculations verify the statement.
\end{proof}

\begin{all}\label{bistoch}
The matrix $\{c_{kl} \}_{k,l \in H} \in \mathbf{M}_{9}(\R)$ defined by equation (\ref{ckldef}) is orthogonal, because $\Theta$ and $M$ are orthogonal matrices. From this follows that the Hadamard-square of this matrix, $\{c_{kl}^2 \}_{k,l \in H}$ is bistochastic. That is $\sum_{k \in H} c_{kl}^2=1 \ \forall l \in H,$ $\sum_{l \in H} c_{kl}^2=1 \ \forall k \in H.$
\end{all}

\begin{tetel}\label{ch_mx}
\be \label{alths}
f(\underline{\lambda},\underline{0},\underline{\tau},\underline{\vartheta},N) \geq \frac{1}{N}\left(6-(\lambda_1^2+\lambda_2^2+\lambda_3^2) \right),
\ee
and (\ref{alths}) holds with equality, if $\underline{\tau} =\underline{\vartheta}=\underline{0}.$
\end{tetel}

\begin{proof}
The distance of the channel matrix and its estimation:
\[
\|A(\lambda_1, \lambda_2, \lambda_3, \phi_z, \phi_y, \phi_x )-A(\hat \lambda_1, \hat \lambda_2, \hat \lambda_3, \hat \phi_z, \hat \phi_y, \hat \phi_x )\|^2=\|A(\lambda_1, \lambda_2, \lambda_3, \phi_z, \phi_y, \phi_x )-\hat A_s\|^2=
\]
\[
=(\hat a_{11}-a_{11})^2+(\hat a_{22}-a_{22})^2+(\hat a_{33}-a_{33})^2+2 \left(\frac{\hat a_{12}- a_{12}}{2}+ \frac{\hat a_{21}- a_{21}}{2} \right)^2+
\]
\[
+2 \left(\frac{\hat a_{13}- a_{13}}{2}+ \frac{\hat a_{31}- a_{31}}{2} \right)^2+
2 \left(\frac{\hat a_{23}- a_{23}}{2}+ \frac{\hat a_{32}- a_{32}}{2} \right)^2.
\]
From Lemma \ref{unbias1} the mean squared error of the $\hat a_{ij}$-s can be written as their variance, hence
\begin{align*}
&f(\underline{\lambda},\underline{0},\underline{\tau},\underline{\vartheta},N)=\Ev \left( \|A(\lambda_1, \lambda_2, \lambda_3, \phi_z, \phi_y, \phi_x )-A(\hat \lambda_1, \hat \lambda_2, \hat \lambda_3, \hat \phi_z, \hat \phi_y, \hat \phi_x )\|^2 \right)=
\\
&=\frac{1}{2}\bigg( \Var (\hat a_{12}+\hat a_{21})+\Var (\hat a_{13}+\hat a_{31})+\Var (\hat a_{23}+\hat a_{32})\bigg)+\sum_{i \in \{1,2,3\}} \Var (\hat a_{ii})=
\\
&=-\frac{1}{2}\bigg( \Var (\hat a_{12}-\hat a_{21})+\Var (\hat a_{13}-\hat a_{31})+\Var (\hat a_{23}-\hat a_{32})\bigg )+\sum_{i,j \in \{1,2,3\}} \Var (\hat a_{ij}).
\end{align*}

By Lemma \ref{linear}
\[
\Var(\hat a_k)=\sum_{l \in H} c_{kl}^{2} (\underline{\tau},\underline{\vartheta}) \frac{1-x_l^2}{N},
\]
hence using the bistochastic property of the matrix $\{c_{kl}^2 \}_{k,l \in H}$ described in Proposition \ref{bistoch}, we can see that
\be \label{sumvarhatak}
\sum_{k \in H} \Var(\hat a_k)= \sum_{k \in H} \sum_{l \in H} c_{kl}^{2} (\underline{\tau},\underline{\vartheta}) \frac{1-x_l^2}{N}=
\sum_{l \in H} \frac{1-x_l^2}{N}=\frac{1}{N} \left( 9 - \sum_{l \in H} x_l^2\right).
\ee

On the other hand, Lemma \ref{linear} shows that
\be
\Var \left(\hat a_m - \hat a_n \right)= \sum_{l \in H} (c_{ml}-c_{nl})^{2} \frac{1-x_l^2}{N}.
\ee
From the orthogonality of $\{c_{kl} \}_{k,l \in H}$ follows that 
$\sum_{l \in H} (c_{ml}-c_{nl})^{2}=2,$
since this expression is the norm-square of orthogonal unit vectors in $\R^9$. 
Therefore $\frac{1-x_l^2}{N} \leq \frac{1}{N}$ implies that $ \Var \left(\hat a_m - \hat a_n \right) \leq \frac{2}{N}.$  

Finally
\be \label{xlambda2}
\sum_{l \in H} x_l^2=\mathrm{Tr} X X^{T}=\mathrm{Tr} (M^{-1} A \Theta) (\Theta^{-1} A^T M)= \mathrm{Tr} A A^T= \lambda_1^{2}+\lambda_2^{2}+\lambda_3^{2}
\ee
hence
\[
f(\underline{\lambda},\underline{0},\underline{\tau},\underline{\vartheta},N) \geq -\frac{1}{2} \left( \frac{2}{N}+\frac{2}{N}+\frac{2}{N}\right)+ \frac{1}{N}\left(9-(\lambda_1^2+\lambda_2^{2}+\lambda_3^{2}) \right).
\]
which is equivalent to the inequality (\ref{alths}).
It is easy to check that in the  $\underline{\tau}=\underline{\vartheta}=\underline{0}$ case 
\begin{align*}
&\Var(\hat a_{ij})=\frac{1}{N}(1-\delta_{ij} \lambda_i^2), ~i,j \in \{1,2,3\}, 
\\
&\Var (\hat a_{12}-\hat a_{21})=\Var (\hat a_{13}-\hat a_{31})=\Var (\hat a_{23}-\hat a_{32})=\frac{2}{N},
\end{align*}
hence the minimum of $f(\underline{\lambda},\underline{0},\underline{\tau},\underline{\vartheta},N)$ is taken in $\underline{\tau}=\underline{\vartheta}=\underline{0}.$
\end{proof}


\subsection{The optimal estimation of the contraction parameters}

It is difficult to compute the loss functions $g$ and $h$ defined in (\ref{eff2}) and (\ref{eff1}),  because the $\lambda_i \ (i \in \{1,2,3\})$ and the $\phi_\alpha \ (\alpha \in \{z,y,x\})$ parameter estimators are non-linear. However, one can approximate these parameter estimators with the their first-order Taylor polynomial (recall that the $\hat \lambda_i, \hat \phi_\alpha$ estimators are $\mathbf{M}_3 (\R) \rightarrow \R$ functions).
Lemma \ref{unbias1} shows that  $\Ev \left( \hat a_{ij}\right)=a_{ij}$, while from Lemma \ref{linear} it follows that $\Var \left( \hat a_{ij}\right)= \underline{\underline{O}}\left( \frac{1}{N}\right)$. Hence the Taylor polynomial of functions $\hat \lambda_i, \hat \phi_\alpha$ with the base point $\Ev (\hat A)=A$ will be an appropriate approximation, if $N$ is large enough.
\par
So the linearized estimators can be written in the form

\begin{align}
 \label{lambdatildedef}
\tilde \lambda_i &:=\hat \lambda_{i}(A)+\inner{\mathrm{grad} \hat \lambda_{i}(A) }{\hat A-A} \ (\forall \ i \in \{1,2,3\}),
\\
\label{fitildedef}
\tilde \phi_{\alpha} &:=\hat \phi_{\alpha}(A)+\inner{\mathrm{grad} \hat \phi_{\alpha}(A) }{\hat A-A} \ (\forall \ \alpha \in \{z,y,x\}).
\end{align}

\begin{lem}[Unbiasedness 2]
The above defined $\tilde \lambda_i \ (i \in \{1,2,3\}), \tilde \phi_{\alpha} \ (\alpha \in \{z,y,x\})$ estimations of the channel parameters are unbiased.
\end{lem}
\begin{proof}
It follows from (\ref{balinv2}) that we have $\hat \lambda_{i} \left(A(\underline{\lambda}, \underline{\phi})\right)=\lambda_i$ and $\hat \phi_\alpha \left(A(\underline{\lambda}, \underline{\phi})\right)=\phi_\alpha.$ Using the unbiasedness of $\hat A$ we get
\[
\Ev (\tilde \lambda_i)=\Ev(\hat \lambda_{i}(A))+\Ev\left(\inner{\mathrm{grad} \hat \lambda_{i}(A) }{\hat A-A} \right)=\hat \lambda_{i}(A)+\inner{\mathrm{grad} \hat \lambda_{i}(A) }{\Ev(\hat A-A)}=
\]
\[
=\hat \lambda_{i}(A)+0=\lambda_{i},
\]
and similarly
\[
\Ev (\tilde \phi_\alpha)=\Ev(\hat \phi_\alpha (A))+\Ev\left(\inner{\mathrm{grad} \hat \phi_\alpha (A) }{\hat A-A} \right)=\phi_\alpha.
\]
\end{proof}


According to Proposition \ref{rotation}, it is enough to consider the case when the Pauli channel described by $A$ has angle parameters $\phi_z=\phi_y=\phi_x=0.$

Let us define the mapping T with formula
\[
T: \mathbf{M}_3 (\R) \rightarrow \R^6; \ \hat A \mapsto (\hat \lambda_1, \hat \lambda_2, \hat \lambda_3, \hat \phi_z, \hat \phi_y, \hat \phi_x)
\]

\begin{lem}\label{deriv}
The non-vanishing components of $\mathrm{d}T\left(A(\underline{\lambda},\underline{0})\right)$ are
\[
\frac{\pd  \hat \lambda_1}{\pd \hat a_{11}}=\frac{\pd  \hat \lambda_2}{\pd \hat a_{22}}=\frac{\pd  \hat \lambda_3}{\pd \hat a_{33}}=1,  \ \frac{\pd  \hat \phi_z}{\pd \hat a_{12,s} }=\frac{1}{\lambda_1-\lambda_2}, \ \frac{\pd  \hat \phi_y}{\pd \hat a_{13,s} }=\frac{1}{\lambda_1-\lambda_3}, \ \frac{\pd  \hat \phi_x}{\pd \hat a_{23,s} }=\frac{1}{\lambda_2-\lambda_3},
\]
where $\hat a_{ij,s}=\frac{1}{2} \left( \hat a_{ij}+ \hat a_{ji} \right).$
\end{lem}

\begin{proof}
 The
parameter estimation $T$ is the left-inverse of the channel parametrization
\[
A: \mathcal{D} \rightarrow \mathbf{M}_3 (\R); \ (\lambda_1,\lambda_2, \lambda_3, \phi_z, \phi_y, \phi_x) \mapsto A(\lambda_1,\lambda_2, \lambda_3, \phi_z, \phi_y, \phi_x),
\]
hence
\be
\mathrm{d}T\left(A(\underline{\lambda},\underline{\phi})\right)= \left( \mathrm{d}A(\underline{\lambda},\underline{\phi})\right)^{-1}.
\ee
We can calculate the elements of $\mathrm{d}A(\underline{\lambda},\underline{\phi})$ easily. For example: 
\[
A( \lambda_1,\lambda_2, \lambda_3,\phi_z, 0, 0)=
\left(
\begin{array}{ccc}
\lambda_1\cos^2{\phi_z}+\lambda_2 \sin^2{\phi_z} & (\lambda_1-\lambda_2)\sin{\phi_z}\cos{\phi_z} & 0\\
(\lambda_1-\lambda_2)\sin{\phi_z}\cos{\phi_z} & \lambda_1\sin^2{\phi_z}+\lambda_2 \cos^2{\phi_z} & 0 \\
0 & 0 & 1
\end{array}
\right),
\]
so
\[
\frac{\pd  A}{\pd \phi_z }( \lambda_1,\lambda_2, \lambda_3,0, 0, 0)=
\left(
\begin{array}{ccc}
0 & (\lambda_1-\lambda_2)& 0\\
(\lambda_1-\lambda_2) & 0 & 0 \\
0 & 0 & 0
\end{array}
\right).
\]
From similar calculations we get that
\[
\mathrm{d}A(\lambda_1,\lambda_2, \lambda_3,0, 0, 0)= \mathrm{Diag}(1,1,1, \lambda_1-\lambda_2, \lambda_1-\lambda_3, \lambda_2-\lambda_3).
\]
from which the statement follows.
\end{proof}

 Let us substitute the result of Lemma \ref{deriv} into the definition of $\tilde \lambda_i$ 
(see (\ref{lambdatildedef})
),
 and use the unbiasedness of $\tilde \lambda_i, 
$ and $\hat A.$ Then the loss function defined by the \emph{linearized} parameter estimators will have the following form:

\begin{align}
\nonumber \tilde g(\underline{\lambda},\underline{0},\underline{\tau},\underline{\vartheta},N):&=\Ev \left((\tilde \lambda_1- \lambda_1)^2+(\tilde \lambda_2- \lambda_2)^2+(\tilde \lambda_3- \lambda_3)^2 \right)=
\\
&=\Var (\hat a_{11})+\Var (\hat a_{22})+\Var (\hat a_{33}).
\end{align}

This function is asymptotically equal to $g$, but easier to handle, thus, we can perform the analytical optimization.

\begin{tetel}\label{contraction}
\be \label{altkont}
\tilde g(\underline{\lambda},\underline{0},\underline{\tau},\underline{\vartheta},N) \geq \frac{1}{N}\left(3-(\lambda_1^2+\lambda_2^2+\lambda_3^2)\right),
\ee
and (\ref{altkont}) holds with equality, if $\underline{\tau} =\underline{\vartheta}=0.$
\end{tetel}

\begin{proof}

\be \label{varf2}
\tilde g(\underline{\lambda},\underline{0},\underline{\tau},\underline{\vartheta},N)=\sum_{i \in \{1,2,3\}} \Var (\hat a_{ii})
=\sum_{i,j \in \{1,2,3\}} \Var (\hat a_{ij})-\sum_{i \neq j} \Var (\hat a_{ij}).
\ee
(\ref{11lem}) shows that $\Var \left(\hat a_{ij}\right) \leq \frac{1}{N} \ (\forall i,j),$ since $\frac{1-x_{ij}^2}{N} \leq \frac{1}{N}$ and the $\{c_{kl}^2 \}_{k,l \in H}$ matrix is bistochastic. Therefore
\[
\sum_{i \neq j} \Var (\hat a_{ij}) \leq \frac{6}{N}.
\]
Now, by the results of (\ref{sumvarhatak}) and (\ref{xlambda2}) we can write the following inequality based on the equation (\ref{varf2}):
\[
\Var (\hat a_{11})+\Var (\hat a_{22})+\Var (\hat a_{33}) \geq \frac{1}{N}\left(9-(\lambda_1^2+\lambda_2^2+\lambda_3^2) \right)- \frac{6}{N},
\]
and this is the statement of the theorem.
It is easy to see that in the  $\underline{\tau}=\underline{\vartheta}=\underline{0}$ case $\Var(\hat a_{11})=\frac{1}{N}(1-\lambda_1^2), \ \Var(\hat a_{22})=\frac{1}{N}(1-\lambda_2^2), \ \Var(\hat a_{33})=\frac{1}{N}(1-\lambda_3^2),$ hence $\tilde g(\underline{\lambda},\underline{0},\underline{\tau},\underline{\vartheta},N)$ is minimal in $\underline{\tau} =\underline{\vartheta}=0.$
\end{proof}

\subsection{The optimal estimation of the angle parameters}

From Lemma \ref{deriv} we can determine the loss function that measures the accuracy of the estimation of the angle parameters:
\[
\tilde h(\underline{\lambda},\underline{0},\underline{\tau},\underline{\vartheta},N)=\Ev \left((\tilde \phi_z-\phi_z)^2+(\tilde \phi_y-\phi_y)^2+(\tilde \phi_x-\phi_x)^2 \right)=
\]
\[
=\frac{1}{4(\lambda_1-\lambda_2)^2} \Var (\hat a_{12}+\hat a_{21})+\frac{1}{4(\lambda_1-\lambda_3)^2} \Var (\hat a_{13}+\hat a_{31})+\frac{1}{4(\lambda_2-\lambda_3)^2} \Var (\hat a_{23}+\hat a_{32}).
\]
Using the result of Lemma \ref{linear}, $ \Var (\hat a_{12}+\hat a_{21}),$ $\Var (\hat a_{13}+\hat a_{31})$ and $\Var (\hat a_{23}+\hat a_{32})$ can be expressed as a function of $\underline{\tau}, \underline{\vartheta}$ and $X=\{x_l\}_{l \in H}.$
The definition of matrices $\Theta$ and $M$ and equations (\ref{csathat}) and (\ref{xmeghat}) show that $X$ can be written as a function of $\underline{\tau}, \underline{\vartheta}$ and $\underline{\lambda}$ (recall that we fixed $\underline{\phi}=\underline{0}$).
Hence, $\tilde h(\underline{\lambda},\underline{0},\underline{\tau},\underline{\vartheta},N)$ can be written in a quite extensive, but explicit, closed form.

For fixed $\underline{\lambda}$ and $N,$ the optimization of $\tilde h(\underline{\lambda},\underline{0},\underline{\tau},\underline{\vartheta},N)$ is a minimization problem with six variables ($\tau_z, \tau_y, \tau_x,  \vartheta_z, \vartheta_y, \vartheta_x$), but unfortunately this optimization problem can not be solved analytically. \par
However, we can formulate conjectures based on numerical optimization computations.


\begin{sejtes}\label{conj}
For any fixed parameters $\underline{\lambda}, N$ 
\bit
\item if $\tilde h(\underline{\lambda},\underline{0},\underline{\tau},\underline{\vartheta},N)$ is minimal at $(\underline{\tau}_{opt},\underline{\vartheta}_{opt}),$ then $\underline{\tau}_{opt}=\underline{\vartheta}_{opt},$
\item the estimation strategies described by parameters $\underline{\tau}_1=\underline{\vartheta}_1=(\frac{\pi}{4}, \frac{\pi}{4},0)$ and $\underline{\tau}_2=\underline{\vartheta}_2=(\frac{\pi}{4},0, \frac{\pi}{4})$ are \emph{nearly} optimal.
\eit
\end{sejtes}

The following examples taken from our numerical optimizations studies 
indicate the validity of Conjecture \ref{conj}. 
\begin{pelda}
Let us fix the following parameters: $\lambda_1=0.8, \lambda_2=0.65, \lambda_3=0.5$ and $N=1000$. 
Then the optimal $\left(\underline{\tau}_{opt},\underline{\vartheta}_{opt}\right)$ (we can calculate them numerically) does not show any regularity except $\underline{\tau}_{opt}=\underline{\vartheta}_{opt}$, with
\[
\min_{\underline{\tau}, \ \underline{\vartheta}} \tilde h=\tilde h (\underline{\tau}_{opt},\underline{\vartheta}_{opt})=0.03634.
\]
We can calculate the values at the two points given in Conjecture \ref{conj}:
$ \tilde h (\underline{\tau}_1,\underline{\vartheta}_1)=\tilde h (\underline{\tau}_2,\underline{\vartheta}_2)=0.03676$.
So the difference can be considered small compared to $\tilde h(\underline{0},\underline{0})=0.05.$
\end{pelda}

\begin{pelda}
The situation is similar for the parameters $\lambda_1=0.9, \lambda_2=0.67, \lambda_3=0.6$ and $N=1000$.
For the optimal input and measurement directions we have $\underline{\tau}_{opt}=\underline{\vartheta}_{opt}$ from numerical optimization, with
\[
\min_{\underline{\tau}, \ \underline{\vartheta}} \tilde h=\tilde h (\underline{\tau}_{opt},\underline{\vartheta}_{opt})=0.01659.
\]
In comparison, $ \tilde h (\underline{\tau}_1,\underline{\vartheta}_1)=\tilde h (\underline{\tau}_2,\underline{\vartheta}_2)=0.01675$ and  $\tilde h (\underline{0},\underline{0})=0.02446.$
\end{pelda}

\paragraph*{Pauli channel with known parameters in one direction}
Let us assume that we have some information about the Pauli channel:
\be
v_3=\sigma_3 \text{ and } \lambda_3=0.
\ee
In this case the channel matrix has the following simplified form
\be
A(\lambda_1,\lambda_2, \phi)=R(\phi) \Lambda(\lambda_1,\lambda_2) R(\phi)^{-1}.
\ee
Now its is easy to show that the input states and the von Neumann measurements should be orthogonal vectors in the plane spanned by $\sigma_1$ and $\sigma_2$. Hence, the input states and the measurements can be parametrized using a single angle parameter $\vartheta$ and $\tau$, respectively:
\[ \Theta=\left[
\begin{array}{ccc}
\theta_{11} & \theta_{21}& 0\\
\theta_{12} & \theta_{22} & 0 \\
0 & 0 & 1
\end{array}\right]=R(\vartheta), \quad M=\left[ \begin{array}{ccc}
m_{11} & m_{21}& 0\\
m_{12} & m_{22} & 0 \\
0 & 0 & 1
\end{array}\right]=R(\tau).
\]

\par

Then we have to solve the following minimization problem with two variables  ($\tau, \vartheta$) for fixed $\lambda_1, \lambda_2,N$ values:
\be
\tilde h_2(\lambda_1, \lambda_2,0, \tau, \vartheta, N)=\Ev \left(\tilde \phi-\phi \right)^2.
\ee
It can be solved analytically, but the proof is rather technical than difficult. The result is formulated in the following theorem:
\begin{tetel}
Assume that $\lambda_1> \lambda_2$ and $\lambda_1 \neq - \lambda_2.$
\ben
\item If $(\lambda_1+\lambda_2)^2 \geq 2 (\lambda_1-\lambda_2)^2,$
then $\tilde h_2(\lambda_1, \lambda_2, 0, \tau, \vartheta, N)$ is minimal if and only if 
\[
\tau_{opt}=\vartheta_{opt}=\frac{\pi}{4} \quad  \left(\mathrm{mod} \ \frac{\pi}{2}\right).
\] The minimal value of the loss function is
\[
\tilde h_2(\lambda_1, \lambda_2, 0, \frac{\pi}{4}, \frac{\pi}{4}, N)=\frac{1}{4(\lambda_1-\lambda_2)^2} \frac{1}{2N}\left(4-(\lambda_1+\lambda_2)^2\right).
\]
\item If $(\lambda_1+\lambda_2)^2 < 2 (\lambda_1-\lambda_2)^2,$
then $\tilde h_2(\lambda_1, \lambda_2, 0, \tau, \vartheta, N)$ is minimal if and only if
\[
\tau_{opt}=\vartheta_{opt}=x \textrm{ or } \frac{\pi}{2}- x \quad \left(\mathrm{mod} \ \frac{\pi}{2}\right),
\]
where $x=\frac{1}{4}\arccos{\left(-\frac{(\lambda_1+\lambda_2)^2}{2 (\lambda_1-\lambda_2)^2}\right)}$.
The minimal value is now  
\[
\tilde h_2(\lambda_1, \lambda_2, 0, \tau_{opt}, \vartheta_{opt}, N)=\frac{1}{4(\lambda_1-\lambda_2)^2} \frac{1}{2N}\left(4-(\lambda_1^2+\lambda_2^2)-\frac{1}{8} \frac{(\lambda_1+\lambda_2)^4}{(\lambda_1-\lambda_2)^2}\right).
\]
\een
\end{tetel}

\begin{pelda}
If $\lambda_1=0.8$ and $\lambda_2=0.2$ the optimal input and measurement directions are $\tau_{opt}=\vartheta_{opt}=\frac{\pi}{4}$. These are the optimal angles in most cases, however, if $\lambda_1=1$ and $\lambda_2=0$ then $\tau_{opt}=\vartheta_{opt}=\frac{\pi}{3} \textrm{ or } \frac{\pi}{6}$.
\end{pelda}

\section{Conclusion and discussion}

In Theorem \ref{alap_tetel} we gave a parametrization of the channel matrix $A$, while in Theorem \ref{domain} we obtained the domain where the parametrization is bijective. We gave an estimation procedure for the channel matrix and channel parameters in Section 3, and then optimized this process with respect to input qubit and measurement directions. We have proven that for the optimal estimation of a channel matrix, we have to take input qubit -- measurement pairs in the channel directions (Theorem \ref{ch_mx}), and the same statement is true for estimating the channel contraction parameters (Theorem \ref{contraction}). This result can be implemented using a two-step algorithm for optimally estimating the channel parameters, where in the first, shorter step we make a rough estimation on the channel directions, and then we set these directions for the second step of the algorithm when the contraction parameters are obtained. 

The estimation of the angle parameters can not be performed analytically in the general case, and the loss function has a quite complex form even in a simplified case. From simulation investigations we conjecture that using a complementary basis to the channel directions would give a nearly optimal result.

\end{document}